\documentclass[a4paper,11pt]{article}

\usepackage{graphicx}
\usepackage[font=small]{caption}
\usepackage{hyperref}
\usepackage{xcolor}
\usepackage{amsmath,amsfonts,amssymb}
\usepackage{amsthm}
\usepackage{bbm}
\usepackage{authblk}
\usepackage{fullpage}

\newcommand{\keywords}[1]{\noindent\textbf{\textit{Keywords---}} #1}

\newcommand{\C}{\ensuremath{\mathbb{C}}}

\newcommand{\Z}{\ensuremath{\mathbb{Z}}}
\newcommand{\R}{\ensuremath{\mathbb{R}}}

\newcommand{\ones}{\mathbbm{1}\hspace*{-0.5mm}}
\newcommand{\Ann}{\ensuremath{{\rm{Ann}}}}

\newcommand{\Per}{\ensuremath{{\rm{Per}}}}
\newcommand{\per}{\ensuremath{{\rm{Per}}}}

\newcommand{\supp}{\ensuremath{{\rm{Supp}}}}
\newcommand{\A}{\ensuremath{\mathcal{A}}}
\newcommand{\inner}[2]{{\langle #1, #2 \rangle}}

\newcommand{\Patt}[2]{\ensuremath{\mathcal{L}_{#2}(#1)}}
\newcommand{\Lang}[1]{\ensuremath{\mathcal{L}(#1)}}
\renewcommand{\vec}[1]{\mathbf{#1}}
\newcommand{\XP}{\ensuremath{{\cal X}}}
\newcommand{\ocl}[1]{\ensuremath{ \overline{{\cal  O}(#1)}} }

\newtheorem{remark}{Remark}
\newtheorem{theorem}{Theorem}
\newtheorem{lemma}{Lemma}
\newtheorem{corollary}{Corollary}
\theoremstyle{definition}
\newtheorem{exm}{Example}

\newenvironment{example}{%
  \begin{exm}
}{%
  \end{exm} \begin{center}\rule{2cm}{0.4pt}\end{center}
}

\newenvironment{thmbis}
  {\addtocounter{theorem}{-1}%
   \begin{theorem}}
  {\end{theorem}}

\begin{document}

\title{Expansivity and periodicity in algebraic subshifts}

\author{Jarkko Kari\\ {\tt jkari@utu.fi} }

\affil{Department of Mathematics and Statistics, University of Turku, Finland}

\maketitle

\abstract{
\begin{center}
\begin{minipage}{\dimexpr\paperwidth-10cm}
\noindent A $d$-dimensional configuration $c:\Z^d\longrightarrow A$ is a coloring of the $d$-dimensional infinite grid by elements of a finite alphabet $A\subseteq\Z$. The configuration $c$ has an annihilator
if a non-trivial linear combination of finitely many translations of $c$ is the zero configuration. Writing $c$ as a $d$-variate formal power series, the annihilator
is conveniently expressed as a $d$-variate Laurent polynomial $f$ whose formal
product with $c$ is the zero power series. More generally, if the formal product is a strongly periodic configuration, we call the polynomial $f$ a periodizer of $c$.
A common annihilator (periodizer) of a set of configurations is called an annihilator (periodizer, respectively) of the set. In particular, we consider annihilators and periodizers of $d$-dimensional subshifts, that is, sets of configurations defined by disallowing some local patterns. We show that a $(d-1)$-dimensional linear
subspace $S\subseteq\R^d$ is expansive for a subshift if the subshift has a periodizer whose support contains
exactly one element of $S$. As a subshift is known to be finite if all
$(d-1)$-dimensional subspaces are expansive, we obtain a simple necessary condition on
the periodizers that guarantees finiteness of a subshift or, equivalently, strong periodicity of a configuration. We provide examples in terms of tilings of $\Z^d$
by translations of a single tile.
 \end{minipage}
 \end{center}
}
\bigskip


\section{Introduction}\label{sec1}

A configuration  in this paper is a coloring of the $d$-dimensional grid $\Z^d$ using finitely many colors. Our colors are integers.
A configuration $c$ has an annihilator if the zero configuration can be obtained as a non-trivial
linear combination of suitable translations of $c$.
In other terms, annihilation means that a linear cellular automaton
maps the configuration $c$ to the zero configuration.
This mapping, in the terminology of digital signal processing, is filtering
by a $d$-dimensional discrete-time
finite-extend impulse response (FIR) filter. Writing $c$ as a $d$-variate formal power series, the annihilator
is conveniently expressed as a $d$-variate Laurent polynomial $f$ whose formal
product with $c$ is the zero power series.

Configurations that have annihilators come up in several contexts. Every low-complexity configuration has an annihilator, where low-complexity means
that the number of patterns in the configuration of some finite fixed shape $D\subseteq\Z^d$ is at most the size $\lvert D\rvert$ of the shape~\cite{icalp}. Low-complexity configurations
are the object of interest in the unsolved Nivat's conjecture~\cite{nivat}, and also in the recently solved periodic tiling problem~\cite{greenfeld-tao2022} where
tilings of $\Z^d$ by translates of a single tile are low-complexity configurations~\cite{icalp}. Also so-called perfect colorings of grid graphs have annihilators~\cite{DLT}.

Configurations with annihilators have global rigidity, although they are not necessarily periodic.
In the two-dimensional case, periodicity in all directions
is known to be enforced if the annihilator
has no line polynomial factors, that is, an annihilating polynomial does not have a non-monomial factor whose monomials are on a single line~\cite{fullproofs,surveyjarkko,DLT}. In this paper we present a similar
condition that works in all dimensions $d$. More generally, we provide a
condition on the annihilator that enforces expansivity: this is a directional
determinism property studied in multidimensional symbolic dynamics.
Expansivity in all directions is known to imply strong periodicity~\cite{Boyle-Lind}.

The article is organized as follows. In Section~\ref{sec:preliminaries} we present
necessary terminology, our notations and some results we need from literature.
In Section~\ref{sec:tiling} we discuss a particular application: tilings
of $\Z^d$ by translated copies of a single tile. Throughout the article, we
demonstrate our methods with examples that come from this setup.
Section~\ref{sec:expansivity} contains the new contributions. We prove a
condition on annihilators that guarantees expansivity, and consequently obtain a
condition that implies strong periodicity of configurations. We provide several examples,
including a discussion on the relation to the Golomb-Welch conjecture. We fisnish with some concluding remarks in Section~\ref{sec:conclusion}.


\section{Preliminaries}
\label{sec:preliminaries}

We start by defining the necessary terminology and concepts. This part is included
for the convenience of the reader although it greatly repeats what is written, for example, in~\cite{surveyjarkko}.

\subsection*{Configurations and periodicity}

A $d$-dimensional \emph{configuration}  over a finite alphabet $A$ is an assignment $$c:\Z^d\longrightarrow A$$
of symbols of $A$ on the infinite grid $\Z^d$.
For any configuration $c\in A^{\Z^d}$ and any cell $\vec{u}\in\Z^d$, we denote by $c_{\vec{u}}$ the letter $c(\vec{u})$
that $c$ has in the cell $\vec{u}$.

 For a vector $\vec{t}\in\Z^d$,  the \emph{translation} $\tau^{\vec{t}}$
shifts a configuration $c$ so that the cell $\vec{t}$ is moved to the cell $\vec{0}$, that is,  $\tau^{\vec{t}}(c)_{\vec{u}}=c_{\vec{u}+\vec{t}}$ for all $\vec{u}\in\Z^d$.
We say that $c$ is \emph{periodic} if $\tau^{\vec{t}}(c)=c$ for some non-zero $\vec{t}\in\Z^d$. In this case $\vec{t}$ is a \emph{vector of periodicity} and $c$ is also called
\emph{$\vec{t}$-periodic}. If there are $d$ linearly independent vectors of
periodicity (viewed as elements of the vector space $\R^d$) then $c$ is called \emph{strongly periodic}.
We denote by $\vec{e}_i=(0,\dots,0,1,0\dots, 0)$ the basic $i$'th unit coordinate vector, for $i=1,\dots ,d$.
A strongly periodic $c\in A^{\Z^d}$ has automatically, for some $k>0$, vectors of periodicity $k\vec{e}_1, k\vec{e}_2, \dots ,k\vec{e}_d$ in the $d$ coordinate directions.

\subsection*{Patterns and pattern complexity}

Let $D\subseteq\Z^d$ be a finite set of cells, a \emph{shape}. A \emph{$D$-pattern} is an assignment $p\in A^D$  of symbols in the shape $D$. A
\emph{(finite) pattern} is a $D$-pattern for some shape $D$. We call $D$ the \emph{domain} of the pattern. Notation $A^*$ is used for the set of all finite patterns over the alphabet $A$ (where the dimension $d$ is assumed to be known).

We say that a finite pattern $p$ of shape $D$ \emph{appears} in a configuration $c$ if for some $\vec{t}\in\Z^d$
we have $\tau^{\vec{t}}(c)\restriction_{D}=p$. We also say that $c$ \emph{contains} the pattern $p$ in the position $\vec{t}$. For a fixed $D$,
the set of $D$-patterns that appear in a configuration $c$ is denoted by $\Patt{c}{D}$. We denote by $\Lang{c}$ the set of all finite patterns that appear in $c$, i.e., the union of
$\Patt{c}{D}$ over all finite $D\subseteq\Z^d$.

The \emph{pattern complexity} of a configuration $c$ with respect to a shape $D$ is the number of different $D$-patterns that $c$ contains. A sufficiently low pattern
complexity forces global regularities in a configuration. A relevant threshold happens when the pattern complexity is at most $\lvert D\rvert$, the number of cells
in shape $D$. Hence we say that $c$ has \emph{low complexity} with respect to shape $D$ if
$$\lvert\Patt{c}{D}\rvert\leq \lvert D\rvert.$$
We call $c$ a \emph{low complexity configuration} if it has low complexity with respect to some finite shape $D$.

\subsection*{Subshifts}

Let $p\in A^D$ be a finite pattern of a shape $D$.
The set $[p]=\{c\in A^{\Z^d}\ \mid\ c\restriction_{D}=p\}$ of configurations that have $p$ in the domain $D$ is called
the \emph{cylinder} determined by $p$.
The collection of cylinders $[p]$ is a base of a compact topology on $A^{\Z^d}$, the \emph{prodiscrete} topology.
See, for example, the first few pages of~\cite{tullio} for details.
The topology is equivalently defined by a metric on $A^{\Z^d}$ where two configurations are close to each other if they agree
with each other on a large region around the cell $\vec{0}$. Cylinders are clopen in the topology: they are both open and closed.

A subset $X$ of $A^{\Z^d}$ is called a \emph{subshift} if it is closed in the topology
and closed under translations.
Note that -- somewhat nonstandardly --  we allow $X$ to be the empty set.
By a compactness argument one has that every configuration $c$ that is not in $X$ contains a finite pattern $p$
that prevents it from being in $X$: no configuration that contains $p$ is in $X$. We can then as well  define subshifts  using forbidden patterns:
given a set $P$ of finite patterns we define
$$\XP_P=\{c\in A^{\Z^d}\ \mid\  \Lang{c}\cap P=\emptyset\},$$
the set of configurations that do not contain any of the patterns in $P$.
The set $\XP_P$ is a subshift, and every subshift is $\XP_P$ for some $P$.
If $X=\XP_P$ for some finite $P$ then $X$ is a \emph{subshift of finite type} (SFT).


For a subshift $X\subseteq A^{\Z^d}$ (or actually for any set $X$ of configurations) we define
its language $\Lang{X}\subseteq A^*$ to be the set of all finite patterns that appear in some element of $X$, that is, the union of
sets $\Lang{c}$ over all $c\in X$. For a fixed shape $D$, we analogously define
$\Patt{X}{D}=\Lang{X}\cap A^D$, the union of all $\Patt{c}{D}$ over $c\in X$. We say that $X$
has low complexity with respect to shape $D$ if $\lvert\Patt{X}{D}\rvert\leq \lvert D\rvert$.
For example, if we fix shape $D$ and a small set
$P\subseteq A^D$ of at most $\lvert D\rvert$ allowed patterns of shape $D$,
then $X=X_{A^D\setminus P}=\{c\in A^{\Z^d}\ \mid\ \Patt{c}{D}\subseteq P \}$
is a low complexity SFT since $\Patt{X}{D}\subseteq P$ and $\lvert P\rvert\leq \lvert D\rvert$.

The \emph{orbit} of a configuration $c$ is the set ${\cal  O}(c) = \{\tau^{\vec{t}}(c)\ \mid\ \vec{t}\in\Z^2\ \}$ of all its translates, and
the \emph{orbit closure} $\overline{{\cal  O}(c)}$ of $c$ is the topological closure of its orbit. The orbit closure is a subshift,
and in fact it is the intersection of all subshifts that contain $c$. In terms of finite patters,
$c'\in \overline{{\cal  O}(c)}$ if and only if $\Lang{c'}\subseteq\Lang{c}$. Of course,
the orbit closure of  a low complexity configuration is a low complexity subshift.

\subsection*{Annihilators and periodizers}

To use commutative algebra we assume that $A\subseteq \Z$, i.e., the symbols in the configurations are integers. We also maintain the assumption that $A$ is finite.
We express
a $d$-dimensional configuration $c\in A^{\Z^d}$ as a formal power series  over $d$ variables $x_1,\dots x_d$ where the monomials address
cells in a natural manner $x_1^{u_1}\cdots x_d^{u_d} \longleftrightarrow (u_1,\dots ,u_d)\in\Z^d$, and the coefficients of the monomials in the power series
are the symbols at the corresponding cells.
Using the convenient vector notation $\vec{x}=(x_1,\dots x_d)$ we write
$\vec{x}^{\vec{u}}=x_1^{u_1}\cdots x_d^{u_d}$ for the monomial that represents cell
$\vec{u}=(u_1,\dots u_d)\in\Z^d$. Note that all our power series and polynomials are \emph{Laurent} as we allow negative as well as positive powers of variables.
Now the configuration $c\in\A^{\Z^d}$ can be coded as the formal power series
$$
c(\vec{x}) = \sum_{\vec{u}\in\Z^d} c_{\vec{u}}\vec{x}^{\vec{u}}.
$$
The power series $c(\vec{x})$ is \emph{integral} (the coefficients are integers) and
because $A\subseteq \Z$ is finite, it is \emph{finitary}
(there are only finitely many different coefficients). Henceforth we treat configurations as integral, finitary power series. By default, for any Laurent power series or polynomial $f$ we denote by $f_{\vec{u}}$ the coefficient of $\vec{x}^{\vec{u}}$.

Note that the power series are indeed formal: the role of the variables is only to provide the position information on the grid. We may sum up two power series, or multiply a power series with a polynomial, but we never plug in any values in the variables.
Multiplying a power series $c(\vec{x})$ by a monomial $\vec{x}^{\vec{t}}$ simply adds $\vec{t}$ to the exponents of all monomials, thus producing the
power series of the translated configuration $\tau^{\vec{t}}(c)$. Hence the configuration $c(\vec{x})$
is $\vec{t}$-periodic if and only if $\vec{x}^{\vec{t}}c(\vec{x})=c(\vec{x})$, that is,
if and only if $(\vec{x}^{\vec{t}}-1)c(\vec{x})=0$, the zero power series. Thus we can express the periodicity of a configuration in terms of its \emph{annihilation} under
the multiplication with a \emph{difference binomial} $\vec{x}^{\vec{t}}-1$. Very naturally then we introduce
the \emph{annihilator ideal}
$$\Ann(c) = \{ f\in \C[\vec{x}^{\pm 1}]\ \mid\ fc=0 \}$$
containing all the polynomials that annihilate $c$. Here we use the notation $\C[\vec{x}^{\pm 1}]$ for the set of Laurent polynomials with complex coefficients.
Note that  $\Ann(c)$ is indeed an ideal of the Laurent polynomial ring $\C[\vec{x}^{\pm 1}]$.

Let us denote the \emph{support} of a Laurent polynomial  $f\in \C[\vec{x}^{\pm 1}]$ by
$$
\supp(f)=\{\vec{u}\in\Z^d\ \mid\ f_{\vec{u}}\neq 0\}.
$$
\begin{remark}
\label{remark:r1}
If a configuration $c$ has an annihilator $f$ with complex coefficients then it also
has an annihilator $f'$ with integer coefficients that satisfies
$\supp(f')=\supp(f)$.
\end{remark}

To see why the remark is true, note that
the annihilation condition $fc=0$ can be viewed as
a homogeneous system of linear equations for the coefficients of the annihilating polynomial $f$.
The coefficients
of the variables in the equations come from the configuration $c$ and are hence integers.
It is easy to see that for any (complex valued) solution of a homogeneous linear system with integer coefficients
there is also an integer valued solution with the property that each variable that had a non-zero value in the original complex solution
also has a non-zero value in the new integral solution.
The integral solution provides the coefficients of an integral annihilator $f'$
that satisfies $\supp(f')=\supp(f)$.

We find it sometimes convenient to work with the \emph{periodizer ideal}
$$\Per(c) = \{ f\in \C[\vec{x}^{\pm 1}] \ \mid\ \mbox{ $fc$ is strongly periodic } \}$$
that contains those Laurent polynomials whose product with configuration $c$ is strongly periodic.
Clearly also $\Per(c)$ is
an ideal of the Laurent polynomial ring $\C[\vec{x}^{\pm 1}]$, and we have $\Ann(c)\subseteq\Per(c)$.
Moreover, if $\Per(c)$ contains non-zero polynomials, so does $\Ann(c)$. Indeed,
if $f\in \Per(c)$ then $fc$
is annihilated by $\vec{x}^{\vec{t}}-1$ for any period $\vec{t}$ of the strongly periodic $fc$,
and thus $f(\vec{x})(\vec{x}^{\vec{t}}-1)$ is an annihilator of $c$.


Our first observation relates the low complexity assumption to annihilators.
Namely, it is easy to see using elementary linear algebra that any low complexity configuration has at least some non-trivial annihilators:
\begin{lemma}[\cite{icalp}]
\label{th:low_complexity}
Let $c$ be a low complexity configuration.
Then $\Ann(c)$ contains a non-zero polynomial. More precisely, if $c$ has low complexity with respect to a
shape $D\subseteq\Z^d$ then there is a non-zero $f\in\Per(c)$ with $-\supp(f)\subseteq D$.
\end{lemma}

The minus sign in front of the support of $f$ in the statement of
the lemma comes from the manner the convolutions in the product
$fc$ are computed: For all $\vec{u}\in\Z^d$
$$
(fc)_{\vec{u}}=\sum_{\vec{v}\in\supp(f)} f_{\vec{v}}c_{\vec{u}-\vec{v}},
$$
so that the pattern of shape $-\supp(f)$ in $c$ at position $\vec{u}$ determines
the new value $(fc)_{\vec{u}}$ at position $\vec{u}$. We see the analogous minus sign
also in
other statements in the rest of the article.

One of the main results of~\cite{icalp} states that if a
configuration $c$ is annihilated by a non-zero polynomial (e.g., due to low complexity) then it is automatically annihilated by a product of difference binomials. This result is fundamental to our approach.
\begin{theorem}[\cite{icalp,fullproofs}]
  \label{th:decompo}
   Let $c$ be a configuration and $f\in\Ann(c)$.
  For every  $\vec{u}\in\supp(f)$ there exist
  pairwise linearly independent $\vec{t}_1, \ldots, \vec{t}_m\in\Z^d$ such that
  each $\vec{t}_i$ is parallel to $\vec{u}_i-\vec{u}$ for some $\vec{u}_i\in\supp(f)\setminus\{\vec{u}\}$, and
  \[ (\vec{x}^{\vec{t}_1} - 1) \cdots (\vec{x}^{\vec{t}_m} - 1) \in \Ann(c) .\]
\end{theorem}
In~\cite{icalp} the statement of Theorem~\ref{th:decompo} is given without reference to
elements of $\supp(f)$ but the given proof
provides for an arbitrary $\vec{u}$ in $\supp(f)$ the vectors $\vec{u}_i\in\supp(f)$
as in the statement above.
In Theorem 12 of~\cite{fullproofs} the result is stated in this stronger form. In the present paper
the directions $\vec{u}_i-\vec{u}$ of $\vec{t}_i$ between positions in the support of
 the annihilating polynomial $f$ play a central role. Note also that by Remark~\ref{remark:r1} the annihilating polynomial $f$ does not need to be integral: there always exists one with the same support and with  integer coefficients.

For a subshift $X\subseteq  A^{\Z^d}$, we denote by $\Ann(X)$ the set of Laurent polynomials that annihilate all elements of $X$, and
we call $\Ann(X)$ the annihilator ideal of $X$. Similarly, $\Per(X)$ is the intersection of sets $\Per(c)$ over $c\in X$. All results stated above for $\Ann(c)$ and $\Per(c)$
for a single configuration $c$
work just as well for $\Ann(X)$ and $\Per(X)$ for a subshift $X$, with similar proofs.
 In particular, we have the following subshift variant of Theorem~\ref{th:decompo}.
\begin{thmbis}
  \label{th:decompo2}
  Let $X$ be a subshift and $f\in\Ann(X)$.
  For every  $\vec{u}\in\supp(f)$ there exist
  pairwise linearly independent $\vec{t}_1, \ldots, \vec{t}_m\in\Z^d$ such that
  each $\vec{t}_i$ is parallel to $\vec{u}_i-\vec{u}$ for some $\vec{u}_i\in\supp(f)\setminus\{\vec{u}\}$, and
  \[ (\vec{x}^{\vec{t}_1} - 1) \cdots (\vec{x}^{\vec{t}_m} - 1) \in \Ann(X) .\]
\end{thmbis}

\section{Tilings by translations of a single tile}
\label{sec:tiling}

As a specific setup and a convenient source of examples throughout the article we consider tilings of $\Z^d$ using
translated copies of a single finite shape $D\subseteq\Z^d$. In this context we call $D$
a \emph{tile}.
A tiling by $D$ is expressed as
a binary configuration where symbols $1$ identify the positions where copies of $D$ are placed to fully
cover $\Z^d$
without overlaps. More precisely,  $c\in\{0,1\}^{\Z^d}$ is a \emph{tiling} by $D$ if and only if $c(\vec{x})f_D(\vec{x})=\ones(\vec{x})$ where
$$
f_D(\vec{x})=\sum_{\vec{u}\in D} \vec{x}^{\vec{u}}
$$
is the characteristic polynomial of $D$, and
$$
\ones(\vec{x}) = \sum_{\vec{u}\in \Z^d} \vec{x}^{\vec{u}}
$$
is the uniform configuration of $1$'s.
The polynomial $f_D$ is thus a periodizer of every tiling by $D$.

Let $T_D\subseteq\{0,1\}^{\Z^d}$ be the set of tilings by $D$. Clearly $T_D$ is a low-complexity subshift of finite type: the elements of $T_D$ are exactly the binary configurations whose $(-D)$-patterns have
precisely one occurrence of symbol $1$, and there exist $\lvert-D\rvert$  such patterns in total.

\begin{figure}[t]
\begin{center}
\includegraphics[width=0.6\textwidth]{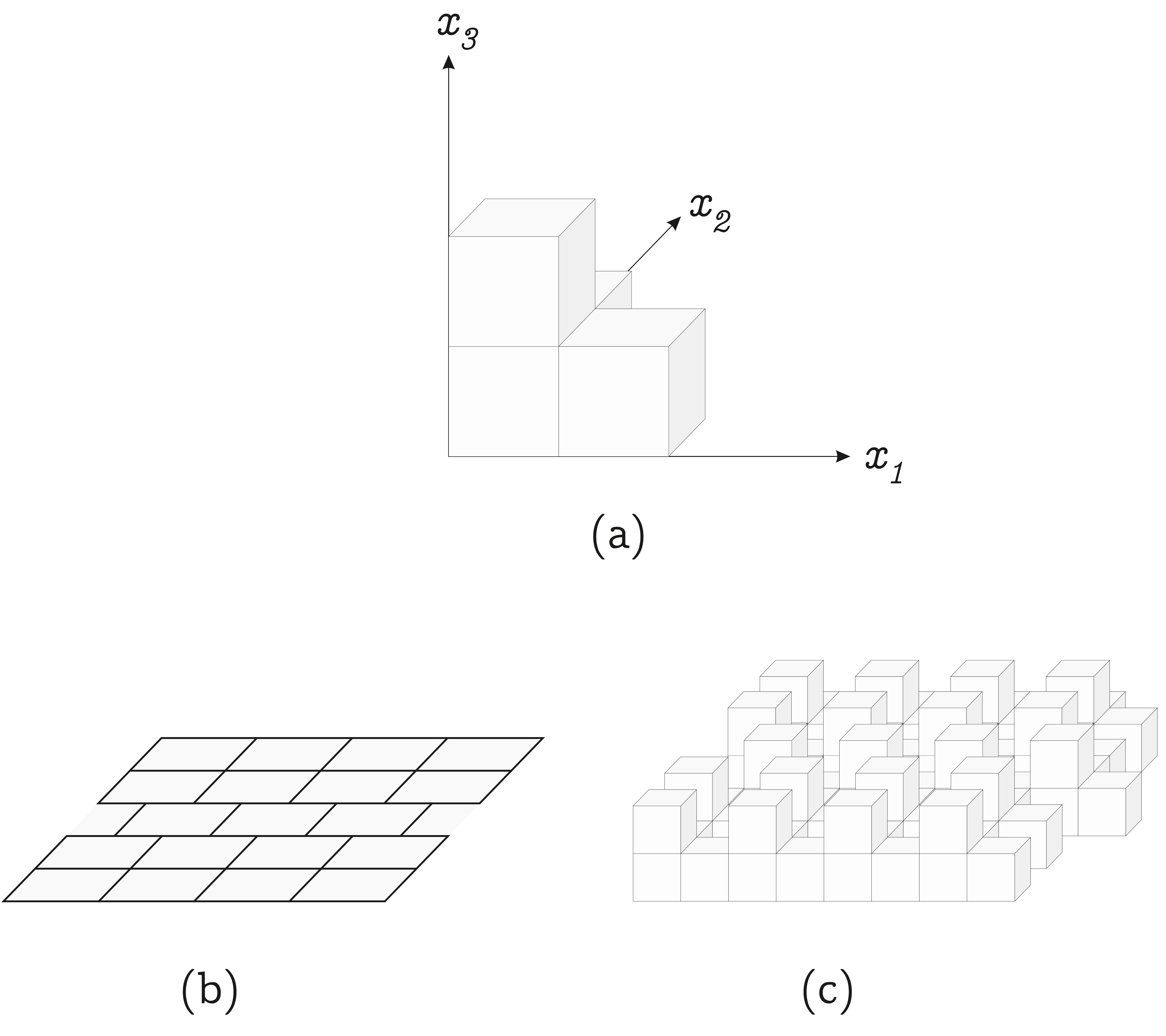}
\end{center}
\caption{(a) The tile $D=\{(0,0,0),(1,0,0),(0,1,0),(0,0,1)\}$ of Example~\ref{ex:ex1},
(b) a tiling of $\Z^2$ by
$2\times 2$ squares that is not strongly periodic since one row of tiles is shifted by one, (c) the corresponding layer of a tiling of $\Z^3$ by $D$. A tiling of $\Z^3$ is obtained by repeating the layer $(1,1,-1)$-periodically.}
\label{fig:fourtile}
\end{figure}

\begin{example}
\label{ex:ex1}
In illustrations we draw tiles in two and three dimensions
as unions of unit squares and cubes. For example,
Figure~\ref{fig:fourtile}(a) shows the tile
$$D=\{(0,0,0),(1,0,0),(0,1,0),(0,0,1)\}.$$
This tile admits tilings of $\Z^3$ that are not strongly periodic.
One may, for example, start with any tiling $a\in\{0,1\}^{\Z^2}$
of $\Z^2$ by the $2\times 2$ square tile $S=\{(0,0),(0,1),(1,0),(1,1)\}$.
Then $c\in \{0,1\}^{\Z^3}$ defined by
$c(x_1,x_2,x_3)=a(x_1+x_3, x_2+x_3)$ is a $(1,1,-1)$-periodic tiling of $\Z^3$ by $D$,
whose slice $(x_1,x_2)\mapsto c(x_1,x_2,0)$ on $\Z\times\Z\times\{0\}$ is
equal to $a$. If $a$ is not strongly periodic then $c$ is not strongly periodic either.
See Figure~\ref{fig:fourtile}(b) and (c).
\end{example}

\begin{example}
\label{ex:ex2}
For a dimension $d$ and radius $r\in\Z_+$, let us denote
$$
B_r^d = \{(n_1,\dots ,n_d)\in\Z^d\ \mid\ \sum_{i=1}^d \lvert n_i\rvert \leq r\}
$$
for the $d$-dimensional radius-$r$ sphere under the Lee metric (also known as the Manhattan metric). See Figure~\ref{fig:leespheres} for
illustrations of $B_2^3$ and $B_3^2$.

If $d\leq 2$ or if $r=1$ then there are strongly periodic
tilings by tile $B_r^d$~\cite{Golomb-Welch-conjecture}: these are perfect codes under the Lee metric. In~\cite{Golomb-Welch-conjecture} it was  conjectured that for other values of $d$ and $r$ the
tile $B_r^d$ does not tile $\Z^d$. There are two natural
variants of the conjecture: the strong  Golomb-Welch conjecture states that no tiling exists, while the
weak Golomb-Welch conjecture postulates that no strongly periodic tiling exists. The conjectures are still open for dimensions $d\geq 6$. It is
known that  the conjecture is true in every dimension for sufficiently large radiuses,
and so the case of radius $r=2$ seems most challenging. See~\cite{Golomb-Welch-survey} for more details.

\end{example}

\begin{figure}[t]
\begin{center}
\includegraphics[width=0.6\textwidth]{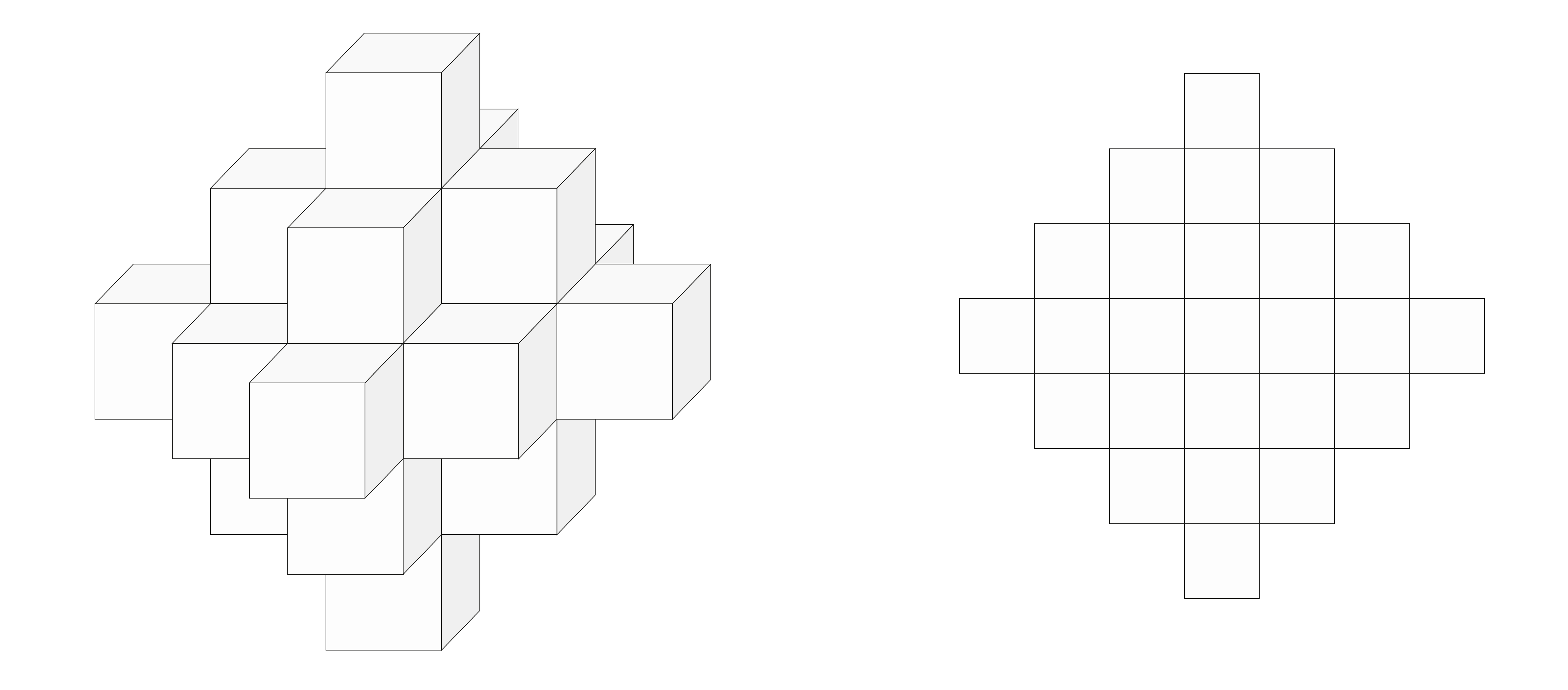}
\end{center}
\caption{The radius-$2$ Lee sphere $B_2^3$ in dimension $d=3$ (on the left), and
the radius-$3$ Lee sphere $B_3^2$ in dimension $d=2$ (on the right).}
\label{fig:leespheres}
\end{figure}

It was recently proved in~\cite{greenfeld-tao2022} that for some dimension $d$, there exists a tile $D\subseteq \Z^d$
such that $T_D$ is an aperiodic SFT, \emph{i.e.}, such that there exists a tiling
but no strongly periodic tiling exists. This provided a negative answer to the Periodic tiling problem~\cite{lagarias-wang}. In contrast, any two-dimensional tile $D\subseteq\Z^2$ that tiles $\Z^2$ also tiles $\Z^2$ periodically~\cite{bhattacharya,greenfeld-tao}.

Interestingly, if $\lvert D\rvert$ is a prime number then every tiling by $D$
is strongly periodic~\cite{Szegedy1998}. This fact has also a simple proof using our algebraic approach, see Example 2 in~\cite{icalp}. In~\cite{Szegedy1998} it was also
shown that $T_D=T_{-D}$ for all tiles $D$, \emph{i.e.}, rotating each tile in place turns a tiling by $D$ into a tiling by $-D$. Thus both $f_D(\vec{x})$ and $f_{-D}(\vec{x})$
are periodizers of valid tilings by $D$.

\section{Expansivity and determinism}
\label{sec:expansivity}

We need some basic concepts of discrete geometry of $\Z^d\subseteq \R^d$.
We use the notation $\inner{\vec{u}}{\vec{v}}$ for the inner product of vectors $\vec{u},\vec{v}\in \R^d$.
For a non-zero vector $\vec{u}\in\R^d$ we denote
$$
H_{\vec{u}}=\{\vec{x}\in\Z^d \mid\ \inner{\vec{x}}{\vec{u}} < 0 \}
$$
for the open discrete half space in the direction $\vec{u}$.
See Figure~\ref{fig:halfplane} for a two-dimensional illustration.

\begin{figure}[h]
\begin{center}
\includegraphics[width=0.4\textwidth]{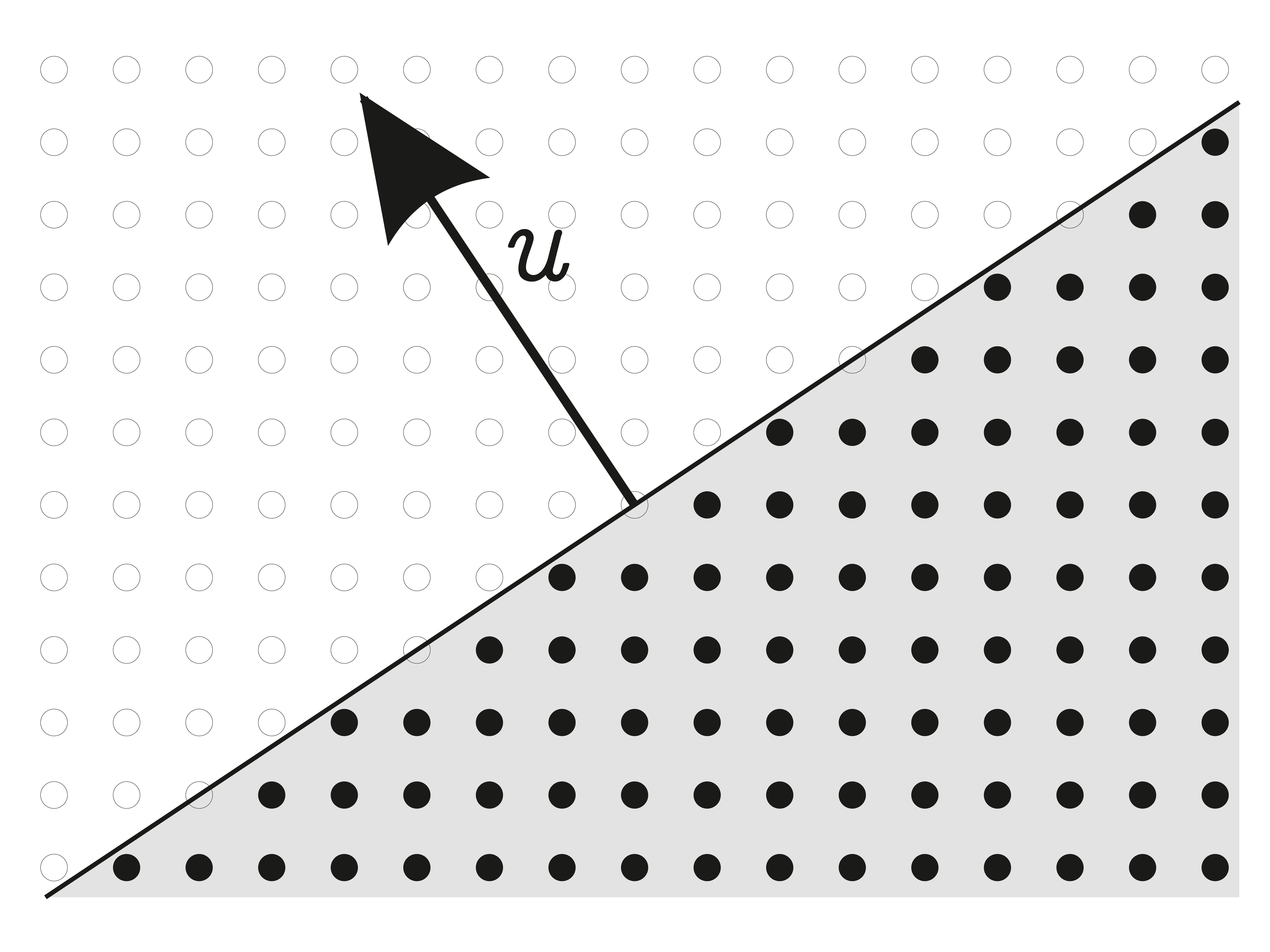}
\end{center}
\caption{The open discrete half space $H_{\vec{u}}$ in dimension $d=2$.}
\label{fig:halfplane}
\end{figure}

A subshift $X$ is \emph{deterministic} in the
direction of $\vec{u}$ if for all $c,c'\in X$
$$
c\restriction_{H_{\vec{u}}}=c'\restriction_{H_{\vec{u}}} \hspace*{5mm}  \Longrightarrow\hspace*{5mm} c=c',
$$
that is, if the contents of a configuration in the discrete half space $H_{\vec{u}}$ uniquely determines the contents in the rest of the cells.
Note that it is enough to verify that the value $c_{\vec{0}}$ on the boundary of the half space is uniquely determined by $c\restriction_{H_{\vec{u}}}$ --- the rest follows by the fact that $X$ is topologically closed and translation invariant.

The following observation is immediate and well known.
It states that if a subshift has as an annihilator (or even as a periodizer) a polynomial $f$ whose negative support $-\supp(f)$
contains a unique position $\vec{v}$ maximally in the direction
of a vector $\vec{u}$ (meaning that the inner product $\inner{\vec{v}}{\vec{u}}$ has maximal value) then $X$ is deterministic in the direction of $\vec{u}$.
In the terminology of~\cite{cyr-kra}, the set $-\supp(f)$ is generating for the subshift,
as knowing all but one symbol of a pattern of shape $-\supp(f)$ in $\Lang{X}$
uniquely identifies also the unknown symbol of the pattern.
In Theorem~\ref{thm:main} we generalize this lemma
to the case where  $-\supp(f)$
contains a position with a unique (but not necessarily maximal)
inner product with $\vec{u}$.

\begin{lemma}
\label{lem:corner}
Let $X$ be a $d$-dimensional  subshift and let $f\in\Per(X)$ be such that $\vec{0}\in\supp(f)$.
Let $\vec{u}\in\R^d$ be a non-zero vector such that $-\supp(f)\setminus\{\vec{0}\}\subseteq H_{\vec{u}}$. Then  $X$ is deterministic in the direction of $\vec{u}$.
\end{lemma}

\begin{proof}
Let $c,c'\in X$ be such that $c\restriction_{H_{\vec{u}}}=c'\restriction_{H_{\vec{u}}}$.
By replacing the polynomial $f(\vec{x})$ by $f(\vec{x})(\vec{x}^{\vec{t}}-1)$ where $\vec{t}\in -H_{\vec{u}}$ is a common period of $f(\vec{x})c(\vec{x})$ and $f(\vec{x})c'(\vec{x})$, we may assume that
$f(\vec{x})\in\Ann(c)$ and $f(\vec{x})\in\Ann(c')$. From
$$
0=(fc)_{\vec{0}} - (fc')_{\vec{0}} = \sum_{\vec{x}\in\supp(f)} f_{\vec{x}} c_{-\vec{x}} - \sum_{\vec{x}\in\supp(f)} f_{\vec{x}} c'_{-\vec{x}}
= f_{\vec{0}} c_{\vec{0}} - f_{\vec{0}} c'_{\vec{0}}
$$
we obtain by dividing with $f_{\vec{0}}\neq 0$ that  $c_{\vec{0}}=c'_{\vec{0}}$.
\end{proof}

If a subshift $X$ is deterministic in directions $\vec{u}$ and $-\vec{u}$ then the
$(d-1)$-dimensional subspace $S=\langle \vec{u}\rangle^{\perp}=\{\vec{v}\in\R^d\ \mid\ \inner{\vec{u}}{\vec{v}}=0\}$ is called an \emph{expansive} space for $X$.
Otherwise it is \emph{non-expansive}.
Using the compactness of $X$ one easily sees that the content of a configuration $c\in X$
within bounded distance from the expansive space $S$ uniquely identifies $c$: There exists $\delta>0$ such that for all $c,c'\in X$,
$$
c\restriction_{B}=c'\restriction_{B} \hspace*{5mm}  \Longrightarrow\hspace*{5mm}  c=c',
$$
where $B=\cup_{\vec{s}\in S} B^{\delta}(\vec{s})$
 and $B^{\delta}(\vec{s})=\{\vec{v}\in\Z^d\ \mid \ \inner{\vec{v}-\vec{s}}{\vec{v}-\vec{s}}<\delta^2\}$ is the ball of radius $\delta$
 around $\vec{s}$ under the usual Euclidean metric. See~\cite{Boyle-Lind} for results concerning expansive spaces of multidimensional subshifts. In particular, the following classical result from~\cite{Boyle-Lind} is central to us, stating that if all $(d-1)$-dimensional subspaces are expansive for a $d$-dimensional
 subshift $X$, then $X$ contains only strongly periodic configurations. This result is our link from deterministic directions to periodicity.

\begin{theorem}[\cite{Boyle-Lind}]
\label{thm:boylelind}
A subshift that is deterministic in every direction is finite, and hence only contains strongly periodic configurations.
\end{theorem}

\begin{example}
Consider a tiling $c\in\{0,1\}^{\Z^3}$ of $\Z^3$ by translations of the tile $$D=(0,0,0),(1,0,0),(0,1,0),(0,0,1)\}$$  from Example~\ref{ex:ex1}, illustrated in Figure~\ref{fig:fourtile}(a).
Suppose that $c$ is $\vec{t}$-periodic for $\vec{t}=k(1,1,1)$ for some $k\in\Z_+$.
Let us prove that $c$ is strongly periodic. Polynomials $f_D(\vec{x})=1+x_1+x_2+x_3$ and
$f_{-D}(\vec{x})=1+x_1^{-1}+x_2^{-1}+x_3^{-1}$, as well as
$\vec{x}^{\vec{t}}-1$ and $\vec{x}^{-\vec{t}}-1$
are periodizers of $c$, and hence they are also in $\Per(X)$ for the orbit closure $X=\overline{{\cal  O}(c)}$ of $c$.
 Lemma~\ref{lem:corner} with the
 periodizers (in fact, annihilators) $\vec{x}^{\vec{t}}-1$ and $\vec{x}^{-\vec{t}}-1$ shows that $X$ is deterministic in every direction $\vec{u}$ that is not perpendicular to $\vec{t}$.
Consider then any non-zero $\vec{u}\perp \vec{t}$, meaning that $\vec{u}=(a,b,c)$ with $a+b+c=0$. If $a=0$ then $\vec{u}=(0,b,-b)$ for $b\neq 0$.
Either $(0,1,0)$ (if $b>0$) or $(0,0,1)$ (if $b<0$) is the unique $\vec{v}\in D$ with the largest inner product with $\vec{u}$. Thus the periodizer
$\vec{x}^{-\vec{v}}f_D(\vec{x})$ of $X$ shows, by Lemma~\ref{lem:corner}, that $X$ is deterministic in the direction $\vec{u}$. Cases $b=0$ and $c=0$ are similar. Finally,
if $a$, $b$ and $c$ are all non-zero then one of them, say $a$, has different sign than the other two. Thus $\vec{v}=(1,0,0)$ is the unique element of $D$ with the maximal or the minimal inner product with $\vec{u}$. Hence $\vec{x}^{-\vec{v}}f_D(\vec{x})$
or $\vec{x}^{-\vec{v}}f_{-D}(\vec{x})$ confirms, by Lemma~\ref{lem:corner}, that $X$ is deterministic in the direction $\vec{u}$. We have shown that $X$ is deterministic in every direction. By Theorem~\ref{thm:boylelind} all elements of $X$, including $c$, are strongly periodic.
\end{example}

%

\subsection*{A sufficient condition for expansivity}

Now we are ready to develop our main tool for establishing expansive spaces of
a subshift with annihilators, and consequently strong periodicity of configurations.
We start by noting how the special annihilator $(\vec{x}^{\vec{t}_1} - 1) \cdots (\vec{x}^{\vec{t}_m} - 1)$ provided by
Theorem~\ref{th:decompo2} gives that
$\langle \vec{u}\rangle^{\perp}$ is expansive for $X$ if $\vec{u}$ is such that
$\inner{\vec{u}}{\vec{t}_i}\neq 0$ for all $i\in\{1,\dots, m\}$.

\begin{lemma}
\label{lem:lemx}
Let $X$ be a $d$-dimensional subshift and $(\vec{x}^{\vec{t}_1} - 1) \cdots (\vec{x}^{\vec{t}_m} - 1)\in\Ann(X)$. For every $(d-1)$-dimensional linear subspace $S\subseteq\R^d$, if $\vec{t}_i\not\in S$ for all $i\in\{1,\dots, m\}$ then $S$
is an expansive space for $X$.
\end{lemma}

\begin{proof}
This is an immediate corollary of Lemma~\ref{lem:corner}. Let $\vec{u}\in\R^d$ be such that $S=\langle \vec{u}\rangle^{\perp}$. Noting that
$\vec{x}^{\vec{t}} - 1 = -\vec{x}^{\vec{t}}(\vec{x}^{-\vec{t}} - 1)$,
we may replace any $\vec{t}_i$ by $-\vec{t}_i$ in the annihilator
$f(\vec{x})=(\vec{x}^{\vec{t}_1} - 1) \cdots (\vec{x}^{\vec{t}_m} - 1)$.

By the assumption, for all $i$ we have that
$\inner{\vec{u}}{\vec{t}_i}\neq 0$. If $\inner{\vec{u}}{\vec{t}_i}<0$ we replace
$\vec{t}_i$ by $-\vec{t}_i$ in the annihilator $f$. So we may assume that
$\inner{\vec{u}}{\vec{t}_i}>0$ for all $i\in\{1,\dots, m\}$. But now the
annihilator $f$
satisfies $\vec{0}\in\supp(f)$ and $-\supp(f)\setminus\{\vec{0}\}\subseteq H_{\vec{u}}$,
so that by Lemma~\ref{lem:corner} the subshift $X$ is deterministic in the direction of $\vec{u}$. Since  $S=\langle -\vec{u}\rangle^{\perp}$ we also have determinism
in the direction of $-\vec{u}$.
\end{proof}

The following theorem states a sufficient condition for expansivity in terms of annihilating (or peridizing) polynomials. It generalizes Lemma~\ref{lem:corner}.

\begin{theorem}
\label{thm:main}
Let $X$ be a $d$-dimensional subshift and let $S$ be a proper
linear subspace of $\R^d$.
If $f\in\per(X)$ is such that
\begin{equation}
\label{eq:eq1}
\supp(f)\cap S=\{\vec{0}\}
\end{equation}
then there exist pairwise linearly independent $\vec{t}_1, \ldots, \vec{t}_m\in\Z^d$ such that $\vec{t}_i\not\in S$  for all $i\in\{1,\dots ,m\}$ and
$(\vec{x}^{\vec{t}_1} - 1) \cdots (\vec{x}^{\vec{t}_m} - 1) \in \Ann(X)$. In particular,
if $S$ is $(d-1)$-dimensional then $S$ is expansive for $X$.
\end{theorem}

\begin{proof}
Let us first prove that there exists $g\in\Ann(X)$  that satisfies
$\supp(g)\cap S=\{\vec{0}\}$, \emph{i.e.}, the same equation (\ref{eq:eq1}) that the periodizer $f$ satisfies.
Set $Y=\{fc\ \mid\ c\in X\}$ is a subshift that only contains strongly periodic configurations. Such a subshift is finite. (This is proved in~\cite{Ballier} for
two-dimensional subshifts of finite type, but the proof directly
generalizes to subshifts in any dimension $d$.)
As the dimension of $S$ is at most $d-1$, some unit coordinate vector $\vec{e}$
is not in $S$. Because $Y$ is a finite set of strongly periodic configurations,
its elements have a common period in the direction of $\vec{e}$.
Multiples of the period are also periods, so that
there are arbitrarily large integers $k$ such that
$(\vec{x}^{k\vec{e}}-1)f(\vec{x})\in \Ann(X)$.
Because $\vec{e}\not\in S$,
for all large enough $k$ the support of $\vec{x}^{k\vec{e}}f(\vec{x})$ has an empty intersection with $S$. Consequently, some $g(\vec{x})=(\vec{x}^{k\vec{e}}-1)f(\vec{x})$
satisfies $\supp(g)\cap S=\{\vec{0}\}$ and $g\in\Ann(X)$.

Applying Theorem~\ref{th:decompo2}  with the annihilator $g$ and $\vec{u}=\vec{0}$ gives the desired special annihilator $(\vec{x}^{\vec{t}_1} - 1) \cdots (\vec{x}^{\vec{t}_m} - 1)$, as $\vec{u}_i-\vec{u}\not \in S$ for
$\vec{u}_i \in \supp(g)\setminus\{\vec{u}\}$.
The last claim now directly follows from Lemma~\ref{lem:lemx}.
\end{proof}

Theorems~\ref{thm:boylelind} and \ref{thm:main} directly give the following tool for forced strong periodicity.

\begin{corollary}
\label{cor:maincor}
Let $X$ be a $d$-dimensional subshift such that for every non-zero $\vec{u}\in \R^d$
there exists $f\in\Per(X)$ and $\vec{v}\in \supp(f)$ such that
$\inner{\vec{v}}{\vec{u}}\neq \inner{\vec{v}'}{\vec{u}}$ for all $\vec{v}'\in \supp(f)\setminus\{\vec{v}\}$. Then $X$ is finite and thus only contains strongly periodic configurations.
\end{corollary}

\begin{proof}
For every $(d-1)$-dimensional subspace $S$ we take $\vec{u}\in\R^d$ such that
$S=\langle\vec{u}\rangle^\perp$. Letting $f$ and $\vec{v}$ be as in the statement of the corollary, we have that $\vec{x}^{-\vec{v}}f(\vec{x})$ is a periodizer of $X$ that satisfies (\ref{eq:eq1}). By Theorem~\ref{thm:main} the subspace $S$ is expansive for $X$. Since $S$ was arbitrary, the claim now follows from Theorem~\ref{thm:boylelind}.
\end{proof}

We can also obtain the following corollary for lower dimensional subspaces.

\begin{corollary}
Let $X$ be a $d$-dimensional  subshift, and let $k\leq d-2$. Suppose that for every
$k$-dimensional linear subspace $S\subseteq \R^d$ there exists
$f\in\Per(X)$ such that $\supp(f)\cap S=\{\vec{0}\}$.
Then there exist
$(k+1)$-dimensional linear
subspaces $S_1,\dots ,S_n$, finitely many, such that every $(d-1)$-dimensional
non-expansive space contains
some $S_i$ as its subspace.
\end{corollary}

\begin{proof}
We use mathematical induction on $k$. The base case $k=0$ is easy: The assumption that
for $S=\{\vec{0}\}$ there exists  $f\in\Per(X)$
such that $\supp(f)\cap S=\{\vec{0}\}$ means that $X$ has a non-zero annihilator.
By Theorem~\ref{th:decompo2} there is a special annihilator $(\vec{x}^{\vec{t}_1} - 1) \cdots (\vec{x}^{\vec{t}_m} - 1)$. By Lemma~\ref{lem:lemx}, a $(d-1)$-dimensional space that does not contain any of the vectors $\vec{t}_i$ is expansive for $X$, so the spaces $S_i=\langle \vec{t}_i\rangle$
for $i\in\{1,\dots ,m\}$ satisfy the claim.

Consider then $k\geq 1$ and suppose the claim is true with $k-1$ in place of $k$. The assumption is that
 for every $k$-dimensional linear subspace $S\subseteq \R^d$ there exists
$f\in\Per(X)$ such that $\supp(f)\cap S=\{\vec{0}\}$. Then the analogous assumption with $k-1$ in place of $k$ holds, so that by the inductive hypothesis there exist
$k$-dimensional linear
subspaces $S_1,\dots ,S_n$ such that every non-expansive space contains
some $S_i$. By the assumption, for every $S_i$ there exists $f_i\in
\Per(X)$ such that $\supp(f_i)\cap S_i=\{\vec{0}\}$. This means,
by Theorem~\ref{thm:main}, that for every $i\in\{1,\dots ,m\}$ the subshift
$X$ has a special annihilator $(\vec{x}^{\vec{t}^{(i)}_1} - 1) \cdots (\vec{x}^{\vec{t}^{(i)}_{m_i}} - 1)$ such that $\vec{t}^{(i)}_j\not\in S_i$
for all $j\in\{1,\dots, m_i\}$. Again,
by Lemma~\ref{lem:lemx}, a $(d-1)$-dimensional space that for some $i$
does not contain any of the
vectors $\vec{t}^{(i)}_j$ for $j\in\{1,\dots, m_i\}$ is expansive for $X$.
We conclude that every non-expansive $(d-1)$-dimensional subspace $S$ contains for some $i\in\{1,\dots n\}$ the $k$-dimensional subspace $S_i$, and for some $j\in\{1,\dots, m_i\}$ the vector $\vec{t}^{(i)}_j$. Consequently, $S$ contains the $(k+1)$-dimensional
subspace generated by $S_i$ and $\vec{t}^{(i)}_j\not\in S_i$.
There are finitely many choices of $i$ and $j$.
\end{proof}

In particular, if a $d$-dimensional subshift $X$ has the property that for every $(d-2)$-dimensional subspace $S$ of $\R^d$  there exists
$f\in\Per(X)$ that satisfies (\ref{eq:eq1}), then all but finitely many $(d-1)$-dimensional spaces are expansive for $X$.

\subsection*{Fibers}

The existence of $f\in\per(X)$ that satisfies the condition (\ref{eq:eq1}) can often
be conveniently given in terms of a linear combinations of ``slices'' of
periodizers parallel to $S$.
Let $S\subseteq\R^d$ be a linear subspace. We call a Laurent polynomial $f$ an $S$-\emph{fiber} if $\supp(f)\subseteq S$. Since products and sums of
$S$-fibers are $S$-fibers, all $S$-fibers form a subring of
the Laurent polynomial ring.

By the \emph{restriction} of a Laurent polynomial $f$ in a subspace $S$ we mean the $S$-fiber
$$
\sum_{\vec{u}\in \supp(f)\cap S} f_{\vec{u}}\vec{x}^{\vec{u}},
$$
and we denote it by $f\restriction S$. Thus the restriction is the sum of the monomials
of $f$ that lie in $S$.
The $S$-fibers of a Laurent polynomial ideal $I$
is the set $I\restriction S$ of all $f\restriction S$ for $f\in I$. The set $I\restriction S$ is an ideal of the ring of $S$-fibers. By the $S$-fibers of a single polynomial $f$ we
mean the restrictions $\vec{x}^{\vec{u}}f(\vec{x}) \restriction S$ over $\vec{u}\in \Z^d$, that is, the ``slices'' of $f$ along translated $S$.

The condition (\ref{eq:eq1})
that $\supp(f)\cap S=\{\vec{0}\}$ for some element $f$ of an ideal $I$
is simply stating that $I\restriction S$ contains the monomial $1$, \emph{i.e.}, it is
the complete $S$-fiber ring.
In practice then, verifying this condition for the periodizer ideal $\Per(X)$ of
a subshift amount to finding a non-zero monomial as a
linear combination of $S$-fibers of various $f\in \Per(X)$.

\begin{figure}[t]
\begin{center}
\includegraphics[width=0.4\textwidth]{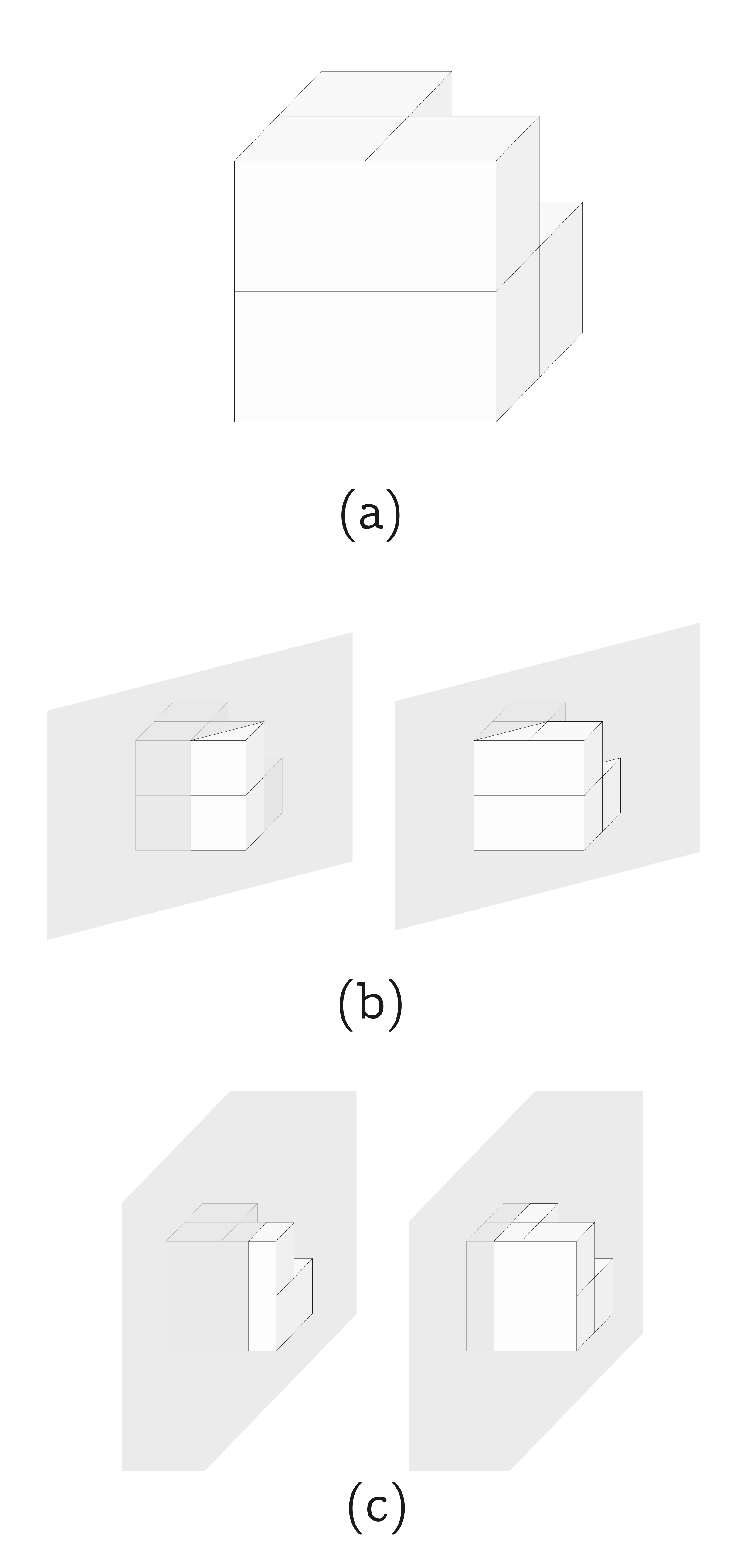}
\end{center}
\caption{(a) The $2\times 2\times 2$ cube missing a corner, studied in Example~\ref{ex:missingcorner}, (b) two fibers parallel to an edge that together
generate a monomial, (c) two fibers parallel to a face that generate a monomial.}
\label{fig:missingcorner}
\end{figure}

\begin{example}
\label{ex:missingcorner}
Let $d=3$ and $D=\{1,\dots ,n_1\}\times \{1,\dots ,n_2\}\times\{1,\dots ,n_3\}\setminus \{(n_1,n_2,n_3)\}$ be a tile for some $n_1,n_2,n_3\geq 2$. The tile is a
rectangular parallelepiped of size $n_1\times n_2\times n_3$  with the missing corner
$(n_1,n_2,n_3)$. See Figure~\ref{fig:missingcorner}(a) for an illustration in the case $n_1=n_2=n_3=2$. Let us prove that every tiling of $\Z^3$ by translations of $D$ is strongly periodic. We prove this by showing that for every two-dimensional linear subspace
$S=\langle\vec{u}\rangle^\perp$ the
$S$-fibers of the periodizer $f_D$ generate a non-zero monomial. Then there is
also a periodizer $f$ that satisfies $\supp(f)\cap S=\{\vec{0}\}$, and we can conclude
strong periodicity using Corollary~\ref{cor:maincor} for $X=\ocl{c}$.

Let $\vec{u}=(a,b,c)$, and consider the following case analysis based on $a$, $b$ and $c$:
\smallskip

\noindent
$\bullet$ If $a\neq 0$, $b\neq 0$ and $c\neq 0$, then one of the corners $(1,1,1)$, $(n_1,1,1)$,
$(1,n_2,1)$ or $(1,1,n_3)$ of $D$ has a unique inner product with $\vec{u}$, and thus provides a monomial $S$-fiber of $f_D$.
\smallskip

\noindent
$\bullet$ If $a\neq 0$ and $b\neq 0$, but $c=0$, then
$f(\vec{x})=1+x_3+x_3^2+\dots +x_3^{n_3}$ is one of the fibers of $f_D$. But there is also a fiber
$g(\vec{x})=1+x_3+x_3^2+\dots +x_3^{n_3-1} + p(\vec{x})f(\vec{x})$ given by the
slice through the missing
corner of $D$, where $p(\vec{x})$ is some polynomial capturing the positions
of full columns on the same plane as the missing corner. See Figure~\ref{fig:missingcorner}(b) for an illustration of this case.
Fibers $f$ and $g$ generate a non-zero monomial
$(1+p(\vec{x}))f(\vec{x}) - g(\vec{x})=x_3^{n_3}$.
The cases when $a=0$ or $b=0$ instead of $c=0$ are symmetric.
\smallskip

\noindent
$\bullet$ Finally, consider the case $a\neq 0$ but $b=0$ and $c=0$. In this case $f_D$ has fibers
$f(\vec{x})=\sum_{i=1}^{n_2}\sum_{j=1}^{n_3}x_2^ix_3^j$ and
$f(\vec{x})-x_2^{n_2}x_3^{n_3}$ whose difference is a non-zero monomial. The
fibers are obtained from slices not containing the missing corner, and containing the missing corner of $D$, respectively.
See Figure~\ref{fig:missingcorner}(c) for
an illustration of this case. Cases where $b\neq 0$ or $c\neq 0$ instead of $a\neq 0$
are symmetric.
\end{example}

\begin{example}
\label{ex:ex4}
Let $D=B_2^d$ be the radius-$2$ Lee sphere in dimension $d\geq 2$, defined in Example~\ref{ex:ex2}. Let us prove that the subspace
$S=\langle(1,1,\dots,1)\rangle^\perp$ is expansive for the subshift $X=T_D$ of
valid tilings of $\Z^d$ by $D$. Note that the direction $\vec{u}=(1,1,\dots,1)$
of determinism is perpendicular to
a $(d-1)$-dimensional discrete facet of $D$, and thus it is
intuitively ``maximally non-deterministic'' among all directions.
We show that a monomial is generated by
two $S$-fibers of $f_D$ corresponding to positions of $D$ having
inner products $0$ and $1$ with $\vec{u}=(1,1,\dots,1)$.
The first fiber, capturing the monomials $\vec{x}^{\vec{v}}$ for $\vec{v}\in D$
with $\inner{\vec{v}}{\vec{u}}=0$ is
$$
f(\vec{x})=1+\sum_{ \begin{array}{c} \scriptstyle  1\leq i,j\leq d\\ \scriptstyle i\neq j\end{array}} x_ix_j^{-1}.
$$
The second fiber, corresponding to positions $\vec{v}\in D$
with $\inner{\vec{v}}{\vec{u}}=1$ is (a monomial multiple of)
$$
g(\vec{x})=\sum_{1\leq i\leq d} x_i.
$$
Because
$$
(x_1^{-1}+x_2^{-1}+\dots x_d^{-1})g(\vec{x}) = f(\vec{x})+(d-1),
$$
we have that the non-zero monomial $d-1$ is an $S$-fiber of $\per(X)$.
\end{example}

\begin{remark}
It remains for future research to determine
whether in the case of Lee spheres $D=B_2^d$
the $S$-fibers of $f_D$ generate a non-zero monomial for all $(d-1)$-dimensional subspaces $S$. If this is the case then
$D=B_2^d$ can only admit strongly periodic tilings, thus proving that the weak and the strong Golomb-Welch conjectures are equivalent for radius-$2$ Lee spheres.
\end{remark}

Note that our methods show that certain subshifts can only contain strongly periodic
configurations. This does not imply that there necessarily are any elements in the subshifts -- the subshift may just as well be empty. For example, it is known
that the Lee sphere $D=B_2^d$ considered in Example~\ref{ex:ex4} does not tile $\Z^d$ in the cases $d\leq 5$, so that in these cases Example~\ref{ex:ex4} concerns
the empty subshift!

\begin{example}
\label{ex:ex8}
Let us continue with the radius-$2$ Lee sphere $D=B_2^3$ in dimension $d=3$,
 illustrated in Figure~\ref{fig:leespheres}. Let us prove that for every
 plane $S=\langle\vec{u}\rangle^\perp$ the $S$-fibers of $f_D$ generate a non-zero monomial. This implies that valid tilings of $\Z^3$
 by $D$ are strongly periodic. However,
 as pointed out above, there are no valid tilings by $D$ so this implication is uninteresting. But the result more broadly implies
 that all configurations $c$ that are periodized by $f_D$, not only the tilings by $D$,
  are strongly periodic.

 Let $\vec{u}=(a,b,c)$. By the symmetries of $D$ we may assume that $a\geq b\geq c\geq 0$. Consider the following case analysis based on $a$, $b$ and $c$:
 \smallskip

 $\bullet$ If $a>b\geq c\geq 0$ then $\vec{v}=(2,0,0)$ is the unique element of $D$ such that $\inner{\vec{v}}{\vec{u}}=2a$, so that $\vec{x}^{\vec{v}}=x_1^2$ provides a monomial $S$-fiber.
  \smallskip

 $\bullet$ If $a=b>c>0$ we take the two $S$-fibers of $f_D$ corresponding to positions of $D$ having inner products $2a$ and $a+c$ with $\vec{u}$. The first fiber, capturing the monomials $\vec{x}^{\vec{v}}$ for $\vec{v}\in D$ with $\inner{\vec{v}}{\vec{u}}=2a$ is
 (a monomial multiple) of
$$f(\vec{x})=x_1^2+x_1x_2+x_2^2,$$
 while the second fiber, corresponding to positions $\vec{v}\in D$
with $\inner{\vec{v}}{\vec{u}}=a+c$ is (a monomial multiple of)
$$g(\vec{x})=x_1x_3+x_2x_3.$$
Their linear combination $f(\vec{x})-x_1x_3^{-1}g(\vec{x})=x_2^2$ is a monomial.
 \smallskip

 $\bullet$ If $a=b>c=0$ then we need three fibers, coresponding to inner product values $2a$, $a$ and $0$. The fibers are (monomial multiples of)
 $$
 \begin{array}{rcl}
 f(\vec{x}) &=& x_1^2+x_1x_2+x_2^2,\\
 g(\vec{x}) &=& x_1+x_2+x_1x_3+x_2x_3+x_1x_3^{-1} + x_2x_3^{-1},\\
 h(\vec{x}) &=& x_3^2+x_3+1+x_3^{-1}+x_3^{-2}+x_1x_2^{-1}+x_1^{-1}x_2.
 \end{array}
 $$
 See Figure~\ref{fig:leestronglyperiodic}. As a linear combination of these we obtain the fiber
 $$
 p(\vec{x}) = x_1^{-2}(1+x_3+x_3^2)f(\vec{x})-x_1^{-2}x_2x_3 g(\vec{x}) = 1+x_3+x_3^2,
 $$
 and then further
 $$
 h(\vec{x})-(1+x_3^{-2})p(\vec{x})-x_1^{-1}x_2^{-1}f(\vec{x}) = -2,
 $$
 a non-zero monomial.
\smallskip

  $\bullet$ The case $a=b=c>0$ was demonstrated in  Example~\ref{ex:ex4}.

\end{example}

\begin{figure}[t]
\begin{center}
\includegraphics[width=0.95\textwidth]{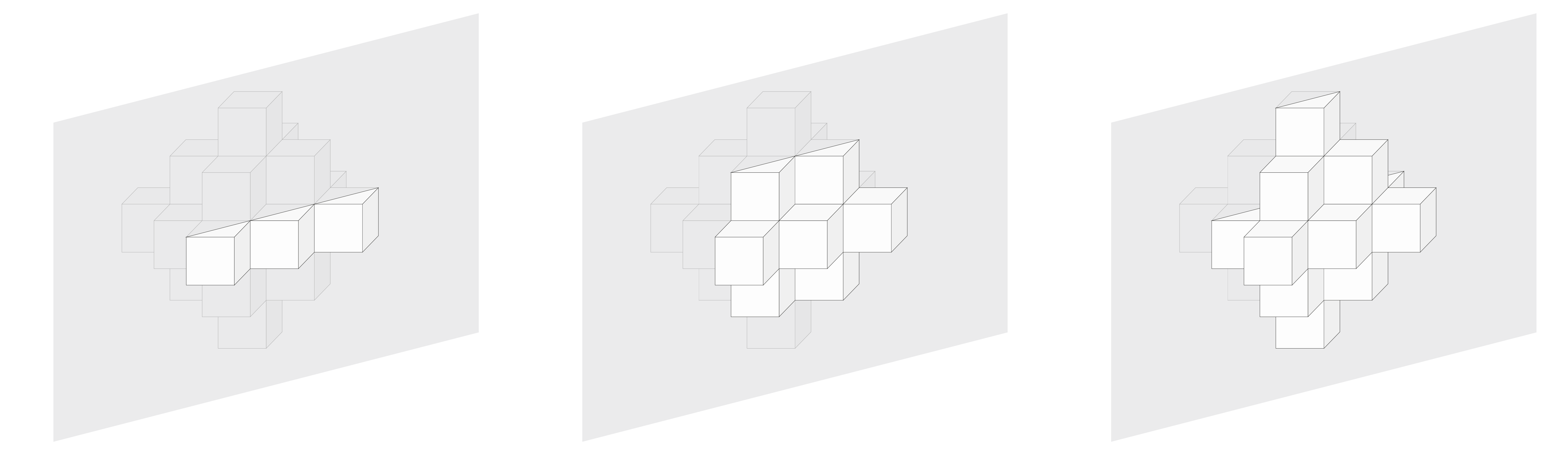}
\end{center}
\caption{Three planes that slice fibers $f$, $g$ and $h$ in the case $a=b>c=0$ of Example~\ref{ex:ex8}.}
\label{fig:leestronglyperiodic}
\end{figure}

Let us finish with some remarks concerning the two-dimensional case $d=2$. In this case our tool to infer strong periodicity of a configuration
is essentially proved in~\cite{fullproofs,surveyjarkko}
using the structure of the annihilator and periodizer ideals. Non-monomial
$S$-fibers for
one-dimensional linear subspaces $S\subseteq\R^2$ are called \emph{line polynomials} as they have at least two monomials and all monomials are along the same line. For any two-dimensional configuration $c$ the periodizer ideal
$\per(c)$ is known to be a principal ideal $\langle\phi_1\phi_2\cdots \phi_m\rangle$
generated by a product of line polynomials $\phi_i$~\cite{fullproofs,surveyjarkko}.
If $c$ has a periodizer $f$
that has no line polynomial factors in any direction then from $f\in \langle\phi_1\phi_2\cdots \phi_m\rangle$ we conclude that $m=0$ so that
$\per(c)=\langle 1\rangle$, implying that $c$ is strongly periodic.
In~\cite{DLT} it was noted that this fact can also be proved without referring to the structure of $\per(c)$ simply by noting that $f$ and the special annihilator
$g(\vec{x})=(\vec{x}^{\vec{t}_1} - 1) \cdots (\vec{x}^{\vec{t}_m} - 1)$ guaranteed by
Theorem~\ref{th:decompo} do not have any common factors as $f$ has no line polynomial factors while all irreducible factors of $g$ are line polynomials. It follows that
there are non-zero
linear combinations of $f$ and $g$ where either one of the two variables has
been eliminated. (These are given by the resultants of $f$ and $g$ with respect to variables $x_1$ and $x_2$, respectively.) Thus there are non-zero annihilators without variables $x_1$ or $x_2$, which implies periodicity of $c$ in horizontal and vertical
directions, \emph{i.e.}, its strong periodicity..

The present paper provides a third proof of this fact that a periodizer without line polynomial factors implies strong periodicity of a two-dimensional configuration.
The present proof has the advantage that it scales to higher dimensions.
One should note, however, that in higher dimensions the statement cannot be given in
terms of $(d-1)$-dimensional
$S$-fibers not having common factors, but rather in terms of
$S$-fibers generating monomial $1$, \emph{i.e.}, generating the full ring.
As line polynomials are essentially one-variate Laurent polynomials, the two conditions are equivalent in the two-dimensional case:
a collection of one-variate Laurent polynomials generate $1$
if and only if the polynomials have no non-trivial common factors. But this is no longer true for
polynomials with two or more variables. (Think of $x-1$ and $y-1$: they have no common factors but as they have a common zero $x=1$, $y=1$, there is no way to express $1$ as their linear combination.)

\section{Conclusion}
\label{sec:conclusion}

We have discussed a method to infer strong periodicity of a multidimensional configuration from its annihilators or periodizers. The method generalizes the two-dimensional technique used in~\cite{fullproofs,surveyjarkko,DLT} to arbitrary dimensions $d>2$. The new method is in fact based on a more general condition on the annihilators or periodizers that implies expansivity of a multidimensional subshift in a given direction. We then use the well known fact that expansivity in all directions implies strong periodicity of the elements of the subshift.

We demonstrated our technique with several examples in the setup of tilings of $\Z^d$ by translated copies of a single tile. The famous Golomb-Welch -conjecture can be stated in this context, and we provided examples related to this conjecture. It remains an interesting topic for future research to see if our method could provide the equivalence of the weak and strong variants of the conjecture, by showing that all tilings by Lee spheres of radius $d\geq 2$ must be strongly periodic. In all the cases that we looked at, it was the case that the fibers extracted from the Lee sphere generated the monomial $1$.

\bibliographystyle{unsrt}

\bibliography{DLT_invited}

\end{document}